\documentclass[a4paper,10pt]{article}
\usepackage{graphicx} % Required for inserting images
\usepackage{amssymb}
\usepackage{enumerate,verbatim}
\usepackage{amsmath}
\usepackage{amsfonts}
\usepackage{amsthm}
\usepackage{tabularx}
\usepackage{xcolor}
\usepackage{epsfig}
\usepackage[normalem]{ulem}

\newtheorem{theorem}{Theorem}
\newtheorem{remark}{Remark}
\newtheorem{definition}{Definition}

\newtheorem{corollary}{Corollary}
\newtheorem{lemma}{Lemma}

\newcommand{\C}{\mathbb{C}}

\def\eqref#1{(\ref{#1})}

\begin{document}
\date{}
\author{Waleed Aziz\\
Department of Mathematics, College of Science\\
Salahaddin University, Erbil, Kurdistan Region, Iraq,
\tt{waleed.aziz@su.edu.krd}\\
{}\\
Colin Christopher\\
School of Engineering, Computing and Mathematics, Plymouth University\\
Plymouth, PL4 8AA, UK,
\tt{C.Christopher@plymouth.ac.uk}\\
{}\\
Chara Pantazi \\
Departament de Matem\`atiques, Universitat Polit\`ecnica de Catalunya (EPSEB)\\
Av. Doctor Maranon, 44-50, 08028, Barcelona, Spain,
\tt{chara.pantazi@upc.edu}\\
{}\\
Sebastian Walcher\\
Fachgruppe Mathematik, RWTH Aachen, 52056 Aachen, Germany\\
\tt{walcher@mathga.rwth-aachen.de}
}
\title{\textbf{Liouvillian integrability of vector fields in higher dimensions}}
\maketitle
\begin{abstract}{We consider complex rational vector fields in dimension $n>2$ (equivalently, differential forms of degree  $n-1$ in $n$ variables) which admit a Liouvillian first integral. Extending a classical result by Singer for $n=2$, our main result states that there exists a first integral which is obtained by two successive integrations from one-forms with coefficients in a finite algebraic extension of the rational function field. The proof uses Puiseux series in a novel way to simplify computations. We also apply this method to give elementary proofs of Singer's theorem for rational one-forms, and of the Prelle-Singer theorem on elementary integrability of rational vector fields.
  
}
\end{abstract}

\section{Introduction and overview of results}

The classical theory of integration in finite terms goes back to Liouville. For 20th century accounts we refer to the seminal works of Risch \cite{Risch1,Risch2} and Rosenlicht \cite{Ros, Ros2}. Liouvillian functions are obtained from rational functions via a finite sequence of adjoining integrals, exponentials and algebraic functions, see \cite{Singer1992, Christopher1999} for details. They play a special role in the integrability problem for functions, vector fields and differential forms. \\
Liouvillian integrability is not only of interest for its own sake but also relevant for applications. There is a number of publications that characterize the Liouvillian first integrals of certain planar families; pars pro toto we just mention Cairo et al.\  \cite{LCHGJL}, Oliveira et al.\ \cite{OST2021}.  We also recall that the existence of a first integral has important consequences for the dynamics of a system; see for example Garc\' ia and Gin\' e \cite{GG2022}. \\
Moreover, one should mention {{some} work that characterizes Liouvillian first integrals of some families in three dimensions; see Ollagnier \cite{OllagnierLiouv, Ollagnier2001}, and some recent studies on integrability aspects of certain three dimensional systems; see Fer\v{c}ec  et al.\ \cite{FRTZ2022}, and also  \cite{LP22023}. \\
Several algorithmic procedures have been presented in the literature to obtain Liouvillian first integrals for two dimensional vector fields. For instance, some of them build on the classical Preller--Singer method \cite{PreSin}. See e.g.\ Avellar et al.\ \cite{ADDM}, Ch\`eze and Combot \cite{CC2020}, Duarte and da Mota \cite{LDLM2}. Concerning algorithms for the computation of Liouvillian first integrals in higher dimensions, see for instance Combot \cite{combot2019}.\\
In an influential paper \cite{Singer1992}, Singer showed that the existence of a Liouvillian first integral of
a two dimensional polynomial vector field is equivalent to the existence of an integrating factor whose logarithmic differential is a closed rational 1-form. As a consequence, if 
$\omega$ is the polynomial 1-form defining the {integral} curves of the
system, there is a closed rational 1-form, $\alpha$, such that  $\omega\cdot\exp(-\int\!\alpha)$ is closed and hence there exists a Liouvillian first integral $\int(\omega\cdot\exp(-\int\!\alpha))$. 
Moreover, such 1-forms are necessarily logarithmic differentials of {\it Darboux} functions, that is, functions of the form
 $\exp{(g/f)}\,\prod f_i^{a_i},$
where the $f_i$, $f$ and $g$ are polynomials in the coordinate variables, and the $a_i$ are complex constants. 
Thus, by Singer's theorem, Darbouxian integrability captures all closed form solutions of two dimensional {Liouvillian integrable} systems. \\
 Singer's theorem has been generalized in various ways, see for example \.{Z}o{\l}\c{a}dek  \cite{Zoladek1998}, Casale \cite{GCasale}, and Zhang  \cite{zhang2016}. In particular, \.{Z}o{\l}\c{a}dek in \cite{Zoladek1998}  presents a multi--dimensional version of Singer's theorem for rational 1-forms.
Zhang \cite{zhang2016} provides a generalization of Singer's theorem to vector fields in $n$ dimensions that admit Darbouxian Jacobi multipliers. \\
 The objective of the present paper is to extend and modify Singer's theorem for complex polynomial or rational vector fields in higher dimensions.  As a preliminary step, we state Singer's theorem for rational one-forms in $n$ variables (due to {Z}o{\l}\c{a}dek \cite{Zoladek1998}) in Theorem \ref{MSing}. Thus, Singer's theorem for 1-forms in dimension two has a natural extension to higher dimensions. The point
of view taken in this paper allows for a very compact proof which we
give in order to motivate the more general case.\\
In addition, we state and prove a characterization of closed 1-forms over the rational function field $K=\mathbb C(x_1,\ldots,x_n)$, as being the 
logarithmic differentials of Darboux functions over $K$.\\
Theorem \ref{algaddendum} is the principal result of the present paper. Informally, it says that there exist a finite algebraic extension $\widetilde K$ of $K$ and 1-forms $\omega$, $\alpha$ over $\widetilde K$ such that $\omega\cdot\exp(-\int\!\alpha)$
is a closed 1-form, and hence there exists a Liovillian first
integral of the form $\int(\omega\cdot\exp(-\int\!\alpha))$.
This is an extension of Singer's theorem for vector fields in $n$ dimensions. In contrast to dimension two, one cannot generally choose $\widetilde K=K$, but the rational function field must be replaced by a finite algebraic extension.\\
In dimension three we furthermore show that if there exists no solution with $\widetilde K=K$, then there exists an inverse Jacobi multiplier over $K$ of Darboux type; see Theorem \ref{Singer}.\\
Our proofs use formal Laurent and Puiseux series throughout. To illustrate the range of applicability of these techniques, we employ them in the final section to re-prove the Prelle-Singer theorem \cite{PreSin} about elementary first integrals for $(n-1)$-forms.

\section{Background and some known results}\label{sbasic}

\subsection{Liouvillian extensions}
We recall some basic notions and facts from differential algebra. For more details see e.g.\ the monograph by Kolchin \cite{Kolchin}. 
Fields are always assumed
to be of characteristic 
zero.

A {\em differential field} is a pair $(K, \Delta)$ where $K$ is a field together with a finite set $\Delta$ of derivations of $K$. Thus for all $\partial\in\Delta$ and all $x,y\in K$ one has the identities
$\partial(x+y)=\partial x + \partial y,\ \partial(xy)=(\partial x) y + x (\partial y).$
We will restrict attention to commutative differential fields, that is the derivations in
$\Delta$ commute. 

The {\em constants} of $(K, \Delta)$ are  those elements $x\in K$ such that
  $\partial x = 0$ for all $\partial\in\Delta$, and the subfield of constants will be denoted by $C_K$.

A {\em differential extension} of $(K, \Delta)$ is a differential field $(\tilde{K}, \tilde{\Delta})$
where $\tilde{K}$ is an extension field of $K$ and each derivation $\tilde{\partial} \in \tilde{\Delta}$
restricts (uniquely) to an element $\partial \in \Delta$. Therefore, it is natural to write $(\tilde{K}, \Delta)$.

We will be mostly interested in the rational function field $\mathbb C(x_1,\ldots,x_n)$, with $\Delta=\{\partial/\partial x_1,\ldots,\partial/\partial x_n\}$, and its extensions.\footnote{One could, in fact, replace $\mathbb C$ by any algebraically closed field of characteristic zero.} Moreover, we focus on Liouvillian extensions (see also Singer \cite{Singer1992}):

\begin{definition}\label{LiouvEx}\textup{
An extension $L\supset K$ of differential fields is called a {\em Liouvillian extension} of $K$ if $C_K = C_L$ and if there exists a tower of fields of the form
%\begin{equation}\label{LiouvExeq} 
$
K=K_0 \subset K_1 \subset \ldots \subset K_m=L,
$
%\end{equation}
such that for each $i\in\{0,\ldots,m-1\}$ we have one of the following:
\begin{list}{}{}
\item{(i)} $K_{i+1}=K_{i}(t_i)$, where $t_i\neq 0$ and $\partial t_i / t_i \in K_{i}$ for all $\partial\in\Delta$; thus $t_i$ is an exponential of an integral of
some element of $K_{i}$.
\item{(ii)} $K_{i+1}=K_{i}(t_i)$, where $\partial t_i \in K_{i}$ for all $\partial\in\Delta$; thus $t_i$ is an integral of an element of $K_{i}$.
\item{(iii)} $K_{i+1}=K_{i}(t_i)$, where $t_i$ is algebraic over $K_{i}$.
\end{list}}
\end{definition}
\begin{remark}{\em 
    By the primitive element theorem, condition (iii) is equivalent to $K_{i+1}$ being a finite algebraic extension of $K_i$. 
    Moreover, we note that the derivations of a given Liouvillian extension $\widetilde K$ of $K$ extend canonically to any finite algebraic extension of $\widetilde K$.}
\end{remark}
We will make extensive use of differential forms, {which generally are more convenient both for the statements and proofs of our results.}
 If $L$ is a differential extension of $K=\mathbb C(x_1,\ldots,x_n)$ then we denote by $L'$ the space of differential 1-forms with coefficients in $L$.
That is, every 1-form $\alpha\in L'$ can be written as $\alpha=\sum a_i \, dx_i$ with $a_i \in L$.
Since the $x_i$ are algebraically independent, we can
treat the $dx_i$ simply as placeholders for the calculations to keep track of the various derivatives. We will freely use the familiar properties of the exterior derivative operator $d$ and of wedge products, but put no deeper algebraic interpretation on the $dx_i$.

Recall that one calls a form $\beta$ {\it closed} whenever $d\beta=0$, and {\it exact} when $\beta =d\theta$ for some form $\theta$.

\begin{remark}\label{LiouvEx2}\textup{
 If $L$ is a differential extension of $K=\mathbb C(x_1,\ldots,x_n)$, then 
  one can restate conditions (i)--(iii) in Definition~\ref{LiouvEx} by the following Types:
\begin{list}{}{}
\item{(i)} $K_{i+1}=K_i(t_i)$, where $t_i\not=0$ and $dt_i=\delta_i t_i$ with some $\delta_i \in K'_i$ (necessarily $d\delta_i=0$).
\item{(ii)} $K_{i+1}=K_i(t_i)$, where $dt_i=\delta_i$ with $\delta_i \in K'_i$  (necessarily $d\delta_i=0$).
\item{(iii)} $K_{i+1}$ is a finite algebraic extension of $K_i$.
\end{list}
{We note that the condition $C_K=C_L$ on constants can always be met in our context for extensions of the rational function field $K$ (see Singer \cite{Singer1992}); so we will not mention it explicitly in the following, freely using the consequence that $d\phi = 0$ for
$\phi \in L$ means that $\phi \in C_K$.}
}\end{remark}

\subsection{Singer's theorem for one-forms}
The following definition is standard. 

\begin{definition}\label{lioudef}\textup{
Given a 1-form $\omega \in \mathbb C(x_1,\ldots,x_n)'$, we say that $\omega$ is {\it Liouvillian integrable} if there exists $\phi$ in some Liouvillian extension $L$ of $\mathbb C(x_1,\ldots,x_n)$
such that $d\phi \wedge \omega = 0$. {More specifically, we will state that $\omega$ is Liouvillian integrable over $L$ when the field of definition is relevant.}}\end{definition}

\begin{remark}\label{liourem}\textup{
\begin{enumerate}[(a)]
    \item We record a convenient characterization of Liouvillian integrability:
    \begin{itemize}
        \item  If the condition in Definition \ref{lioudef} holds, then {$\omega = j\,d\phi$ for some $j \in L$ and $d\omega = \alpha \wedge \omega$ with $\alpha = dj/j$,} thus $d\alpha=0$.
\item Conversely, if there exists a Liouvillian extension $L$ of $K$ and $\omega\not=0,\,\alpha\in L'$ such that $d\omega=\alpha\wedge\omega$ and $d\alpha=0$, then with Remark \ref{LiouvEx2}, part (i), there exists {$j$ in a Liouvillian extension $L_1$ of $L$ such that $\alpha=-dj/j$, whence
\[
d(j\omega)=dj\wedge\omega+j{d}\omega =j(-\alpha\wedge\omega+d\omega)=0,
\]
which in turn implies $j\omega=d\phi$} for some $\phi$ in a Liouvillian extension $L_2\supset L_1$, hence of $\mathbb C(x_1,\ldots,x_n)$.
One calls {$j$} an {\em inverse integrating factor} for $\omega$.
    \end{itemize}
    \item Geometrically, the condition above implies that $\omega=0$ defines a foliation with leaves $\phi={\rm const.}$  {In particular, it is integrable in the sense
    of foliations since $\omega = m\,d\phi$ for some function $m$ and hence $d\omega \wedge \omega = 0$
    (cf.\ e.g.\ Camacho and Lins Neto
    \cite{CLN}, Appendix \S3).
 \item   {When $n = 2$, $\phi$ is a first integral for the two-dimensional vector field $P\frac{\partial}{\partial x}+Q\frac{\partial}{\partial y}$ associated to 
 $\omega = P\,dy-Q\,dx$, which is the application considered
 by Singer.  When $n>2$, the case of interest for us,
 the $n$-dimensional vector field \eqref{vfieldnd} is associated to the
 $(n-1)$-form \eqref{n-1form} which is always integrable
 in the foliation sense.  The existence of a first integral (or $n-1$ first integrals, if we want
 \textit{complete} integrability), is a distinct
 concept.}}
\end{enumerate}
}
\end{remark}

The following result says that Singer's theorem for 1-forms in dimension two carries over to 1-forms in arbitrary dimension. The result, and the first proof, is due to {Z}o{\l}\c{a}dek in \cite{Zoladek1998}. We give a different, elementary, proof here.
\begin{theorem}[Singer's Theorem for 1-forms]\label{MSing}
Let  $\omega$ be a rational 1--form over $K=\mathbb C(x_1,\ldots,x_n)$. Then $\omega$ is Liouvillian integrable if and only if there exists a closed 1--form $\alpha \in \mathbb C(x_1,\ldots,x_n)'$ such that
$d\omega=\alpha \wedge \omega$.
\end{theorem}

\begin{proof}
We proceed by induction on the tower of fields. Let $K_{i+1}$ be a Liouvillian extension of $K_{i}$, of one of the types (i)--(iii) in Definition \ref{LiouvEx}, and 
consider a closed 1--form $\alpha \in K'_{i+1}$ such that
$d\omega=\alpha \wedge \omega$. We have to show that there exists $\tilde{\alpha} \in K'_{i}$ such that
$d\omega=\tilde{\alpha} \wedge \omega$ with $d\tilde{\alpha}=0$.
We discuss the types from Remark~\ref{LiouvEx2} separately.

%\begin{list}{}{}
%\item{Type \bf(i)}. 
\noindent{\bf Type \bf(i).}
We can suppose that $t_i=t$ is transcendental over $K_i$, else this falls into type (iii).
Then ({by Lemma \ref{laurentlem}}) write $\alpha$ as a formal Laurent series in decreasing powers of $t$,
\begin{equation}\label{LAU}
\alpha= \alpha_r t^r+\alpha_{r-1} t^{r-1}+\ldots, \quad \alpha_r \in K'_i, \quad \alpha_r \neq 0.
\end{equation}
Equating powers of $t^0$ in $\alpha \wedge \omega=d\omega$ and $d\alpha = 0$, we see that
$d\omega=\alpha_0 \wedge \omega, \quad d\alpha_0=0.$ Therefore, we can choose $\tilde{\alpha}=\alpha_0 \in K_i$.

%\item{Type \bf(ii)}.
\noindent {\bf {Type \bf(ii).}}
As above, we suppose that $t_i=t$ is transcendental over $K_i$, and write $\alpha$ in the form \eqref{LAU}.
From $d\alpha=0$ we deduce that $d\alpha_r=0$.
Furthermore, from $d\omega=\alpha \wedge \omega$, we obtain three cases
depending on $r$:
\begin{itemize}{}{}
\item If $r>0$, then $\alpha_r \wedge \omega=0$. In this case, there exists $h \in K_i$ such that
$\alpha_r =h \omega$, thus we get $d\omega=-\frac{dh}{h} \wedge \omega$. We may take $\tilde{\alpha}=-\frac{dh}{h}$.
\item If $r=0$, we have $d\omega=\alpha_0 \wedge \omega$ and we may take $\tilde{\alpha}=\alpha_0$.
\item If $r<0$, we see $d\omega=0$ and we may take $\tilde{\alpha}=0$.
\end{itemize}

%\item{Type \bf(iii)}. 
\noindent {\bf Type \bf(iii).}
There is no loss of generality in assuming that the extension is Galois, with Galois group $G$ of order $N$. 
Take traces of both sides of $d\omega=\alpha \wedge \omega$, and of $d\alpha=0$, respectively,
to obtain
\[
d\omega=\left( \frac{1}{N}\sum_{\sigma \in G}{\sigma (\alpha)} \right)\wedge \omega, \qquad
d \left(\frac{1}{N}\sum_{\sigma \in G}{\sigma (\alpha)}\right)=0.
\]
Thus we can choose
$\tilde{\alpha}=\frac{1}{N}\sum_{\sigma \in G}{\sigma (\alpha)} \in K_i$.
%\end{list}
\end{proof}
This proof illustrates
the effectivity of working with power series in the generator of the
extension. For our main result below, however, we will need to use Puiseux expansions.\\
A key role will be played by {\em Darboux functions.}  These are functions of the form 
\begin{equation}\label{D1} \phi=\exp(g/f)\prod f_i^{a_i} ,\end{equation}
where the $f_i$ and $g$ and $f$ are elements of $\C[x_1,\ldots,x_n]$ and $a_i$ are complex numbers. 
Given a Darboux function $\phi$, its logarithmic differential, $d\phi/\phi$, is clearly a closed rational 1-form.
Conversely, we shall show that every closed
rational 1-form must be the logarithmic differential of some Darboux function. The case $n=2$ of the following theorem was given by Christopher \cite{Christopher1999} and, in a different context, by Schinzel \cite{Sch}.

\smallskip
\begin{theorem}\label{CChris}
Consider a 1--form $\alpha\in \mathbb C(x_1,\ldots,x_n)'$. If $\alpha$ is closed, then there exist elements
$g, f, f_i \in \mathbb C[x_1,\ldots,x_n]$ and constants
$a_i \in \C$ such that $$\alpha=d\left(\frac{g}{f}\right)+\sum a_i\frac{df_i}{f_i}.$$
\end{theorem}

\begin{proof}
We proceed by induction on $n$. {The case $n=1$ amounts to the well-known fact that the primitive of a rational function in $x_1$ has the form $r(x_1)+\sum a_i\log(x_1-b_i)$  with $a_i,\,b_i\in\mathbb C$ and a rational function $r$.}\\ Now suppose that $n>1$ and the theorem holds for $\mathbb C(x_1,\ldots,x_{n-1})$.
Let $\bar{K}$ be a splitting field over $\mathbb C(x_1,\ldots,x_{n-1})$ of a common denominator of the coefficients of $\alpha$, { and denote the distinct roots of this common denominator  by $b_1,\ldots,b_r\in\bar{K}$.} 
Then we can write $\alpha$ as a partial fraction expansion in $x_n$ over $\mathbb C(x_1,\ldots,x_{n-1})$:
\[\alpha=\sum^{r}_{i=1}\sum^{n_i}_{j=1}\frac{a_{i,j}}{(x_n-b_i)^j}dx_n+\sum^{N}_{i=0} c_i x^{i}_{n} dx_n+ \sum^{{r}}_{i=1}\sum^{{m}_i}_{j=1}\frac{\Omega_{i,j}}{(x_n-b_i)^j}+\sum^{{M}}_{{i=0}}  x^{i}_{n} {\omega_i},\]
where the $\Omega_{i,j},\, \omega_{i}$ are elements of $\mathbb C(x_1,\ldots,x_{n-1})'$, and
$a_{i,j}$ and $c_i$ are elements of $\mathbb C(x_1,\ldots,x_{n-1})$. 
%Note that,
%we can ignore the limits of all summations without confusion.

 {By evaluating $d\alpha=0$ and comparing coefficients in the partial fraction expansion we get the following for all $i,\,j\geq 0$, where it is understood that $a_{i,0}=0$ and $\Omega_{i,0}=0$:}
\begin{align}
d c_i{-}(i+1)\omega_{i+1}=&0,\label{nC3}\\
d a_{i,{j+1}}{+}j a_{i,j}d b_i{-}j\Omega_{i,j}=&0,\label{nC4}\\
d\omega_i =&0,\label{nC5}\\
d\Omega_{i,j+1}+jd b_i \wedge \Omega_{i,j}=&0.\label{nC6}
\end{align}

{These may be seen as identities in $\mathbb C(x_1,\ldots,x_{n-1})'$. In particular, $d a_{i,1}=0$, so $a_{i,1}\in \mathbb C$.}
From  \eqref{nC5}  $d\omega_0=0$ and hence by hypothesis we can write
\[
\omega_0=d\left(\frac{\tilde{g}}{\tilde{f}}\right)+\sum \tilde{a}_i\frac{d\tilde{f}_i}{\tilde{f}_i},
\]
for some $\tilde{g}, \tilde{f}, \tilde{f}_i \in \mathbb C(x_1,\ldots,x_{n-1})$ and $\tilde{a}_i \in \mathbb C$.
Equations \eqref{nC3} -- \eqref{nC6}  allow us to write 

\begin{equation} \label{nalpha-omega}
\begin{array}{rcl}
    \alpha-\omega_0&=&\displaystyle{\sum\limits_i a_{i,1}\dfrac{d(x_n-b_{i})}{(x_n-b_{i})}+\sum\limits_{j>1}\sum\limits_i d\left(\dfrac{a_{i,j}}{(x_n-b_i)^{j-1}}\left(\frac{-1}{j-1}\right)\right)}\\
     & & +\displaystyle{\sum\limits_i d\left(\dfrac{c_{i}\,x^{i+1}_{n}}{i+1}\right)}.
\end{array}
\end{equation}
Now let $G$ be the Galois group of $\bar{K}$ over $\mathbb C(x_1,\ldots,x_{n-1})$. For any differential form $\mu$ over $\bar{K}$ and $\sigma\in G$ we denote by $\sigma(\mu)$ the form obtained by letting $\sigma$ act on its coefficients. Taking the trace of both sides
of equation \eqref{nalpha-omega}  {and noting that $\sigma$ and the exterior derivative commute}, we have
\begin{equation}
\begin{array}{rcl}
\dfrac{1}{|G|}\displaystyle{\sum\limits_{\sigma \in G}\sigma(\alpha-\omega_0)}&=&\dfrac{1}{|G|}\displaystyle{\sum\limits_{\sigma \in G}\sum a_{i,1}\dfrac{d(x_n-\sigma{(b_{i})})}{(x_n-\sigma{(b_{i})})}}\\
&& +\dfrac{1}{|G|}\displaystyle{\sum\limits_{\sigma \in G}\sum\sum d\left(\dfrac{\sigma{(a_{i,j})}}{(x_n-\sigma{(b_i}))^{j-1}}\left(\dfrac{-1}{j-1}\right)\right)}\\
&&+\dfrac{1}{|G|}\displaystyle{\sum\limits_{\sigma \in G}\sum d \left(\dfrac{\sigma(c_{i})\,x^{i+1}_{n}}{i+1}\right)}.
\end{array}
\end{equation}
Since $G$ is the set of all automorphisms of $\bar{K}$ fixing $\mathbb C(x_1,\ldots,x_{n-1})$, the left hand side of this equation is equal to $\alpha-\omega_0$, and we obtain $\alpha$ in the desired form.
\end{proof}

\begin{remark}\label{LD}\textup{
Combining Theorem~\ref{MSing} and Theorem~\ref{CChris}, we see that a 1-form $\omega$
is Liouvillian integrable if and only if it admits a Darboux integrating factor.
}
\end{remark}

%%%%%%%%%%%%%%%%%%%%%%%%%%%%%%%%%%%%%%%%%%%%%%%%%%%%%%%%%%%%%
\section{Extension of Singer's theorem to vector fields in higher dimensions}

We now consider rational vector fields
\begin{equation}\label{vfieldnd}
\mathcal{X}=\sum_{i=1}^nP_i\frac{\partial }{\partial x_i}
\end{equation}
on $\mathbb C^n$, $n\geq 3$; equivalently the corresponding ($n-1$)-forms
{
\begin{equation}\label{n-1form}
\Omega=\sum_{i=1}^n(-1)^{i-1}P_i\,dx_1\wedge\cdots\wedge\widehat{dx_i}\wedge\cdots\wedge dx_n
\end{equation}
}
defined over $K=\mathbb C(x_1,\ldots,x_n)$.
We will state and prove a weaker version of Singer's theorem that holds for any dimension $n\geq 3$.

\begin{definition}\textup{
A non-constant
element, $\phi$, of a Liouvillian extension of $\mathbb{C}(x_1,\ldots,x_n)$ is called a {\it Liouvillian first integral}  of the vector field $\mathcal{X}$ if it satisfies $\mathcal{X} \phi=0$ or, equivalently, $d\phi \wedge  \Omega=0$. }
\end{definition}

\begin{remark}\textup{ In view of Remark \ref{liourem}, this property is equivalent to the existence of some Liouvillian extension $L$ of $K$ and one-forms $\omega\not=0$, $\alpha$ in $L'$ such that 
\begin{equation}\label{twoformintcond}
    \omega\wedge\Omega=0,\quad d\omega=\alpha\wedge\omega, \quad d\alpha=0.
\end{equation}
To verify the nontrivial part of the equivalence, construct a Liouvillian extension $\widehat L$ of $L$ by successsively adjoining elements with $d\phi/\phi=\alpha$, and $d\psi=\exp(-\int\!\alpha)\omega$, if necessary. Hence,
 we have a first integral of $\Omega$ of the form 
 $\int(\omega\cdot\exp(-\int\!\alpha))$.}\\
{\em In this case we will briefly (and slightly abusing language) say that $\Omega$ is {\em Liouvillian integrable over $L$}}.  
\end{remark}
\subsection{The main result}
Our principal result states that Liouvillian integrability of $\Omega$ over some extension $L$, while not necessarily implying Liouvillian integrability over $K$, does imply Liouvillian integrability over a finite algebraic extension of $K$.
\begin{theorem}[Extension of Singer's theorem to $n$ dimensions]\label{algaddendum}
Let  $\Omega$ be the ($n-1$)--form {\eqref{n-1form} }over $K = \C(x_1,\ldots,x_n)$. If there exists a Liouvillian first integral of $\Omega$, then there exists a finite algebraic extension $\widetilde K$ of $K$, and $\omega, \alpha\in\widetilde K'$, $\omega\not=0$, such that 
\eqref{twoformintcond} holds.
\end{theorem}
Before proving this theorem, we state two lemmas. The proof of the first is straightforward.
\begin{lemma}\label{omegamultilem}
    Let $L$ be a differential extension of $K=\mathbb C(x_1,\ldots,x_n)$, moreover $0\not=\ell\in L$, $0\not=\omega\in L'$ and $\alpha\in L'$ such that 
    $d\omega =\alpha\wedge\omega$, $d\alpha=0$. Then
    \begin{equation}\label{Liouv1}
d\left(\frac{\omega}{\ell}\right)=\left(\alpha-\frac{d \ell}{\ell}\right) \wedge \frac{\omega}{\ell},\quad d\left(\alpha-\frac{d \ell}{\ell}\right) =0.
\end{equation}
\end{lemma}
\begin{lemma}\label{extensions}
    Let $L_0$ be a differential extension of {$K=\mathbb C(x_1,\ldots, x_n)$}, $t$ transcendental over $L_0$ such that $L_0(t)$ is Liouvillian over $L_0$, and $L$ a finite algebraic extension of $L_0(t)$. Assume that $L_0$ and $L$ have the same constants.\\
    If there exist $\omega,\alpha\in L'$ such that $\omega\not=0$, $\omega\wedge\Omega=0$, $d\omega=\alpha\wedge \omega$, $d\alpha=0$ (so $\Omega$ is Liouvillian integrable over $L$),
    then there exists a finite algebraic extension $\widetilde L_0$ of $L_0$, and $\widetilde\omega,\widetilde\alpha\in \widetilde L_0'$ such that $\widetilde\omega\not=0$, $\widetilde\omega\wedge\Omega=0$, $d\widetilde\omega=\widetilde\alpha\wedge \widetilde\omega$, $d\widetilde\alpha=0$ (so $\Omega$ is Liouvillian integrable over $\widetilde L_0$).
\end{lemma}
\begin{proof}
    By Lemma \ref{puislem} (see appendix) there exists a finite extension $\widetilde L_0\supset L_0$ so that we may write $\omega, \,\alpha$
as formal Laurent series in decreasing powers of $\tau= t^{1/m}$ with some positive integer $m$, thus
\begin{equation}
    \omega =\omega_r \tau^r +\omega_{r-1} \tau^{r-1} \ldots, \quad \omega_k \in \widetilde L_0' \  (k\le r), \quad \omega_r \neq 0,
\end{equation}
and either $\alpha=0$ or
\begin{equation}
      \alpha =\alpha_s \tau^s +\alpha_{s-1} \tau^{s-1} \ldots, \quad \alpha_k \in \widetilde L_0'\  (k\le s), \quad \alpha_s \neq 0.
\end{equation}
With $t$ transcendental, we have  $\omega \wedge \Omega=0$, hence $\omega_k \wedge \Omega=0$ for all $k$.\\
We now consider the types of transcendental extensions in Definition \ref{LiouvEx}.
%\begin{itemize}
%\item{Type \bf (i)}.

\noindent{\bf Type  (i).} Let $dt=t\delta$ with $d\delta=0$, hence $
d\tau=\frac1m\tau\delta.
$
We thus obtain the highest degree terms
\begin{equation*}
    d\omega=\tau^r\left(\frac rm\delta\wedge\omega_r+d\omega_r\right)+\cdots,\quad d\alpha=\tau^s\left(\frac sm \delta\wedge\alpha_s+d\alpha_s\right)+\cdots
\end{equation*}
unless $\alpha=0$. Comparing both sides of $\alpha \wedge \omega=d\omega$ yields the following three cases:
\begin{itemize}
\item  When $\alpha=0$ or $s<0$ (thus the highest degree on the left hand side is $<r$), we just have $d\omega_r+\frac rm\, \delta \wedge \omega_r=0$. In this case choose $\widetilde{\alpha}=-\frac rm\, \delta$ (with $d\widetilde{\alpha}=0$)
and $\widetilde{\omega}=\omega_r$.
\item When $s=0$, we see $\alpha_0 \wedge \omega_r=d \omega_r + \frac rm \, \delta \wedge \omega_r$. In this case take $\widetilde{\alpha}=\alpha_0 - \frac rm \, \delta$ (noting $d\alpha_0=0$)
and $\widetilde{\omega}=\omega_r$. 
\item When $s>0$, we get $\alpha_s \wedge \omega_r=0$ and therefore
$\alpha_s= h\, \omega_r$ for some $h\in \widetilde L_0$. Since $\omega_r \wedge \Omega=0$, then
$\alpha_s \wedge \Omega=0$. Moreover $d\alpha_s + \frac sm \, \delta \wedge \alpha_s=0$ from $d\alpha=0$. So we may choose $\widetilde{\alpha}=- \frac sm \, \delta$
(with $d\tilde{\alpha}=0$) and $\widetilde{\omega}=\alpha_s$.
\end{itemize}
%\item{Type \bf (ii)}. 
{\bf Type (ii).} Here we have $t$ transcendental over $\widetilde L_0$, $dt=\delta$, with $d\delta=0$. Therefore $d\tau=\frac{1}{m} \tau^{1-m}\delta$, hence 
\begin{equation*}
    d(\tau^r\omega_r)= \tau^r d\omega_r+\frac rm \tau^{r-m}\delta\wedge\omega_r, 
\end{equation*}
which shows that the leading term of $d\omega$ is just $\tau^r d\omega_r$. Likewise, the leading term of $d\alpha$ equals $\tau^s d\alpha_s$ unless $\alpha=0$.
Comparing the leading terms of $\alpha \wedge \omega=d\omega$,
we obtain three cases:
\begin{itemize}
\item When $s>0$, we get $\alpha_s \wedge \omega_r=0$ and hence $\alpha_s= h\cdot \omega_r$ for some $h\in \widetilde L_0$.
Since $\omega_r \wedge \Omega=0$, then $\alpha_s \wedge \Omega=0$. From
$d\alpha=0$ one sees $d\alpha_s=0$. In this case take $\widetilde{\alpha}=0$ and $\widetilde{\omega}=\alpha_s$.
\item When $s=0$, we see $\alpha_0 \wedge \omega_r=d \omega_r$, and $d\alpha_0=0$ from $d\alpha=0$.
In this case choose $\widetilde{\alpha}=\alpha_0$
and $\widetilde{\omega}=\omega_r$.
\item When {$\alpha=0$ or $s<0$}, then $d \omega_r=0$. Take $\widetilde{\alpha}=0$
and $\widetilde{\omega}=\omega_r$.
\end{itemize}
%\end{itemize}
\end{proof}

The following is now a direct consequence of Lemma \ref{extensions}. 

\begin{proof}[Proof of Theorem \ref{algaddendum}]
    Consider a tower 
    %\begin{equation*} 
    $K=K_0 \subset K_1 \subset \ldots \subset K_m=L, $
    %\end{equation*}
    as in Definition \ref{LiouvEx} (or Remark \ref{LiouvEx2}), and assume that for some $i>1$ there exists a finite algebraic extension $K_{i+1}\supset K_i$, and $\omega$, $\alpha\in K_{i+1}'$ satisfying the conditions in \eqref{twoformintcond}.  With no loss of generality, $K_i\supset K_{i-1}$ is then transcendental, and Lemma \ref{extensions} shows that there exists a finite algebraic extension $\widetilde K_{i-1}$ of $K_{i-1}$ {(by the Primitive Element Theorem, this extension can be generated by one element)}, and $\widetilde\omega,\,\widetilde\alpha\in K_{i-1}'$ as required in \eqref{twoformintcond}.
    Thus all transcendental extensions can be eliminated by descent. 
\end{proof}

Thus, the vector field admits a first integral obtained, via \eqref{twoformintcond}, from integrating the 1-forms, $\omega$
and $\alpha$, defined over $\widetilde{K}$.  That is, 
there is a first integral of the form $\phi=\int{\frac{\omega}{e^{\int{\alpha}}}}$, with $\omega,\,\alpha\in \widetilde{K}'$.  Such an integral will normally be multivalued.

If the conclusion of the theorem holds with $\widetilde K=K$, 
then $e^{\int{\alpha}}$ is of Darboux type by Theorem~\ref{CChris}.  (When $\alpha=0$, we just have a first integral of the form $\int\omega$.)
If there do not exist $\omega$ and $\alpha$ in $K'$ itself that satisfy \eqref{twoformintcond} then we shall call $\Omega$ {\em exceptional}.
An extensive discussion of these exceptional cases will be the subject of a forthcoming paper \cite{CPWLiou2} by Christopher et al.  Here we just note that such cases exist and have a very nice group structure.  In particular, for dimension $n=3$ a correspondence exists between exceptional cases and the finite rotation groups of the sphere.

\subsection{Three dimensional vector fields with Liouvillian first integrals}\label{smain}

In the exceptional cases, further reductions are possible depending on the dimension of the vector field and the size
of the extension $[\widetilde{K}:K]$.   We give some details 
in the particular case when $n=3$, leaving the discussion of higher dimensions to the forthcoming
paper \cite{CPWLiou2}.
Consider the vector field
\begin{equation}\label{vfield3d}
\mathcal{X}=P\frac{\partial }{\partial x}+Q\frac{\partial }{\partial y}+R\frac{\partial }{\partial z}
\end{equation}
in $\mathbb C^3$ with the corresponding 2-forms
\begin{equation}\label{2form}
\Omega=P\,dy \wedge dz+Q\,dz \wedge dx+R\,dx \wedge dy.
\end{equation}

By Theorem \ref{algaddendum} there exists an integral $\omega\in L'$, with $L$ a finite algebraic extension of $K$ with $d\omega = \alpha\wedge\omega$, $\alpha\in L'$, $d\alpha = 0$. Without loss of generality we can assume that $L$ is Galois over $K$ with Galois group $G$. 

Let us first assume that $\sigma(\omega) \wedge \omega=0$ for all $\sigma \in G$. {The vector space of all one-forms $\beta\in L'$ that satisfy $\beta\wedge\Omega=0$ has dimension two over $L$, and one can choose a basis of this subspace with elements of $K'$. Thus there exist} $\eta$, $\theta\in K'$ such that $\eta\wedge\Omega=\theta\wedge\Omega=0$, $\eta\wedge\theta\not=0$, and there exist $k$, $\ell\in L$ such that $\omega=k\eta+\ell\theta$.
With Lemma \ref{omegamultilem} one sees that 
\begin{equation*}
    \widetilde\omega:=\dfrac{\omega}{k}=\eta+\widetilde \ell\theta
\end{equation*}
satisfies $\widetilde\omega\wedge\Omega=0$, and $d\widetilde\omega=\widetilde\alpha\wedge\widetilde\omega$, $d\widetilde\alpha=0$ with $\widetilde\alpha = \alpha - dk/k \in L'$. Since $\sigma(\widetilde\omega)\wedge\widetilde\omega=0$ for all $\sigma\in G$, we have $\sigma(\widetilde\omega)\wedge\widetilde\omega=\left(\ell-\sigma(\ell)\right)\,\eta\wedge\theta$, so that $\sigma(\widetilde\ell)=\widetilde\ell$ for all $\sigma\in G$, and $\widetilde\omega\in K'$. Finally, forming the trace of $d\widetilde\omega=\widetilde\alpha\wedge\widetilde\omega$ shows that one may take $\widetilde\alpha\in K'$.  This case
is not an exceptional one as we can choose $\omega$ and
$\alpha$ to be in $K'$.

Assume now that $\tau(\omega) \wedge \omega \neq 0$ for some $\tau \in G$. In this case there exist two independent Liouvillian first integrals, with differentials $\omega$ and $\tau(\omega)$. Since $\tau(\omega) \wedge \Omega = \omega \wedge
\Omega = 0$, we must have $\tau(\omega) \wedge \omega = \ell\, \Omega$ for some $\ell \in L$.  Taking differentials, we obtain
\[d\Omega = (\alpha+\tau(\alpha)-\frac{d\ell}{\ell})\wedge \Omega,\]
with $d(\alpha+\tau(\alpha)-\frac{d\ell}{\ell}) = 0$.  From the trace of the equation above, we find a $\beta\in K'$ such that $d\Omega = \beta\wedge\Omega$ and $d\beta=0$.  Thus, there
exists an inverse Jacobi multiplier of the form $e^{\int \beta}$.  By Theorem~\ref{CChris}, this must be of Darboux type. We have shown:

\begin{theorem}[Extension of Singer's theorem for three dimensional vector fields]\label{Singer}
Let $K = \C(x,y,z)$, and let $\Omega$ be the 2--form \eqref{2form} over $K$. If there exists a Liouvillian first integral of $\Omega$, then one of the following holds:
\begin{trivlist}{}{}

\item{\textup{(I)}} There exist 1--forms $\omega, \alpha \in K'$ such that
\[
\omega\not=0,\,
\omega \wedge \Omega=0,\: \alpha \wedge \omega=d\omega, \: d\alpha=0.
\]
That is, $\omega$ is Darboux integrable over $K$, and $\Omega$
has a first integral of the form {$\phi=\int\frac{\omega}{j}$ for some Darboux function $j = e^{\int{\alpha}}$} over $K$.
 
\item{\textup{(II)}} There exists a 1--form $\beta \in K'$
such that $d\Omega = \beta \wedge \Omega$ with $d\beta=0$. 
So, $\Omega$ admits an inverse  Jacobi multiplier 
{$j = e^{\int{\beta}}$} of Darboux type over $K=\C(x,y,z)$. \footnote{For the notion of inverse Jacobi multiplier see Berrone and Giacomini \cite{BeGia}. Note that we include nonzero constant functions as multipliers.}  
\end{trivlist}{}{}
\end{theorem}

The proof of the theorem also shows:
\begin{corollary}
Let $K = \C(x,y,z)$, and let $\Omega$ be the 2--form \eqref{2form} over $K$. If $\Omega$ admits a Liouvillian first integral, but not two independent Liouvillian first integrals, then (I) holds.
\end{corollary}

%%%%%%%%%%%%%%%%%%%%%%%%%%%%%%%%%%%%%%%%%%%%%%
\section{A proof of the Prelle-Singer theorem }

In this section we take the Puiseux series approach to prove a well known theorem by Prelle and Singer, thus providing a further illustration of the method.
We consider vector fields \eqref{vfieldnd}, equivalently the 
corresponding forms \eqref{n-1form} of degree $n-1$,
defined over $K=\mathbb C(x_1,\ldots,x_n)$.
We recall some notions and facts (see e.g. Rosenlicht \cite{Ros}).
\begin{definition}\label{elemdef}{\em
\item A differential extension field $L$ of $K$ is called {\it elementary} if and only if $K$ and $L$ have the same constants
 and there exists a tower of fields of the form
\begin{equation}\label{ElemExeq} K=K_0 \subset K_1 \subset \ldots \subset K_N=L, \end{equation}
such that for each $i\in\{0,\ldots,m-1\}$ we have one of the following:
\begin{enumerate}[(i)]
\item $K_{i+1}=K_i(t_i)$, where $t_i\not=0$ and $dt_i/t_i=dR_i $ with some $R_i \in K_i$ (adjoining an exponential);
\item $K_{i+1}=K_i(t_i)$, where $dt_i=dR_i/R_i$ with $R_i \in K_i$ (adjoining a logarithm);
\item $K_{i+1}$ is a finite algebraic extension of $K_i$.
\end{enumerate}
}
\end{definition}

As in the case of Liouvillian extensions, the condition on the constants is unproblematic in our context.

\begin{definition}\label{dLf1}\em{
A non-constant
element, $\phi$, of an elementary extension of $K$ is called an {\it elementary first integral}  of the vector field $\mathcal{X}$ if it satisfies $\mathcal{X} \phi=0$ or, equivalently, $d\phi \wedge  \Omega=0$.}
\end{definition}
The existence of an elementary first integral according to Definition \ref{dLf1} is equivalent to the existence of an elementary extension $L$ of $K$ and {$v\in L$, $u_1,\ldots,u_M\in L\setminus\{0\}$} and $c_1,\ldots,c_M\in \mathbb C$ such that
\begin{equation}\label{elemrel}
    \left(\sum_{i=1}^M c_i\,\dfrac{du_i}{u_i} +dv\right)\wedge\Omega=0, \quad \sum_{i=1}^M c_i\,\dfrac{du_i}{u_i} +dv\not=0.
\end{equation}
To verify the non-obvious implication, adjoin logarithms if needed. \\
Consider the following version of the main theorem in Prelle and Singer \cite{PreSin}, stated for rational ($n-1$)-forms.\footnote{{A recent paper \cite{HZ} by Huang and Zhang extends the Prelle-Singer theorem to $(n-1)$-forms admitting $n-1$ independent elementary first integrals.}}
\begin{theorem}\label{algaddendumel}
Let $K = \C(x_1,\ldots, x_n)$, and let $\Omega$ be the ($n-1$)--form {\eqref{n-1form}} over $K$. If there exists an elementary extension $L$ of $K$ such that a relation \eqref{elemrel} holds with $u_i,\,v\in L$, then there exists a finite algebraic extension $\widetilde L$ of $K$, $\widetilde L\subseteq L$,  such that a relation 
\begin{equation*}
    \left(\sum_{i=1}^{\widetilde M} \widetilde c_i\,\dfrac{d\widetilde u_i}{\widetilde u_i} +d\widetilde v\right)\wedge\Omega=0, \quad \sum_{i=1}^{\widetilde M} \widetilde c_i\,\dfrac{d\widetilde u_i}{\widetilde u_i} +d\widetilde v\not=0.
\end{equation*}
holds with $\widetilde c_i\in \mathbb C$, and $\widetilde u_i,\,\widetilde v\in \widetilde L$.
\end{theorem}
{In our proof of Theorem \ref{algaddendumel} (which is different from the one in \cite{PreSin}) we proceed in a manner similar to our induction proof of Theorem \ref{algaddendum}. For the induction step, the role of Lemma \ref{extensions} is taken over by the following.}

\begin{lemma}\label{extensions2}
    Let $L_0$ be a differential extension of $K$, $t$ transcendental over $L_0$ such that $L_0(t)$ is elementary over $L_0$, and $L$ a finite algebraic
    extension of $L_0(t)$ such that $L$ and $L_0$ have the same constants.\\
    If there exist $v,\,u_i\in L$, and $c_i\in\mathbb C$ such that \eqref{elemrel} holds, then there exists a finite algebraic extension $\widetilde L_0$ of $L_0$, $\widetilde v\in\widetilde L$, {$\widetilde u_i\in \widetilde L\setminus\{0\}$} and $\widetilde c_i\in \mathbb C$ such that
\begin{equation}\label{elemreldown}
    \left(\sum \widetilde c_i\,\dfrac{d\widetilde u_i}{\widetilde u_i} +d\widetilde v\right)\wedge\Omega=0, \quad \sum \widetilde c_i\,\dfrac{d\widetilde u_i}{\widetilde u_i} +d\widetilde v\not=0.
\end{equation}
\end{lemma}
\begin{proof} We consider \eqref{elemrel} over $L$. Given $v$ and the $u_i$, one obtains $\widetilde L_0$ from Lemma \ref{puislem}. We use expansions in descending integer powers of $\tau = t^{1/m}$, frequently using that $\frac{d\tau}{\tau}=\frac1m\frac{dt}{t}$.
Thus
\begin{equation}\label{viexpand}
u_i=\alpha_i\tau^{r_i}+ \text{l.o.t.};\quad r_i\in\mathbb Z;\quad \alpha_i\in\widetilde L_0;
\end{equation}
\begin{equation}\label{uexpand}
v=\sum_j \beta_j\tau^j; \quad j\in \mathbb Z; \quad \beta_j\in \widetilde L_0;
\end{equation}
with the summation in descending order; here ``l.o.t.'' stands for lower order terms.  We may and will assume that $\alpha_i\not=0$. In case $v\not=0$ we let $s$ be the maximal index such that $\beta_s\not=0$.

In Cases (i) and (ii) of Definition \ref{elemdef} above, the differential extends to $\widetilde L_0[[\tau^{-1}]]$ and its quotient field, {and moreover it maps $\tau^{k}\,\widetilde L_0[[\tau^{-1}]]$ to itself, for every integer $\tau$.} Thus we get
\[
\dfrac{du_i}{u_i}\in \widetilde L_0[[\tau^{-1}]]
\]
for all $i$.  Similarly, when $v\not=0$, the maximal index with nonzero coefficient in $dv$ is at most $s$.  We can assume that $t$, and hence
$\tau$, is transcendental over {$L_0$}; otherwise the Lemma is satisfied 
trivially.
\begin{enumerate}
\item We first suppose that $v\not=0$ with $s>0$. 
\begin{itemize}
\item In Case (i) we have
\[
dv=\tau^s\left(d\beta_s+\frac{s}{m}\beta_s\,dR\right) + \text{l.o.t.}
\]
and, comparing coefficients,
\[
\left(\dfrac{d\beta_s}{\beta_s}+ \frac{s}{m}\,dR\right)\wedge\Omega=0.
\]
Here $\dfrac{d\beta_s}{\beta_s}+ \dfrac{s}{m}\,dR\in \widetilde L_0'$ is nonzero, for else
\[
d\tau/\tau= -\frac{1}{s} d\beta_s/\beta_s\quad\text{and hence}\quad \tau=c\cdot \beta_s^{-1/s},
\]
for some constant $c$. This is algebraic over $L_0$; a contradiction. So, when $s>0$ then the assertion holds, and we may assume that $s\leq 0$ in further discussions.
\item In Case (ii) we have
\[
dv=\tau^s\,d\beta_s+ \tau^{s-1}\left(d\beta_{s-1}+\frac{s}{m}\beta_s\,dR/R\right) + \text{l.o.t.},
\]
and evaluation of the integral condition shows $d\beta_s\wedge\Omega=0$. If $\beta_s$ is not constant, then the assertion holds. If $\beta_s$ is constant, and $s>1$, then
\[
\left(d\beta_{s-1}+\frac{s}{m}\beta_s\,dR/R\right)\wedge\Omega=0.
\]
We must have $d\beta_{s-1}+\frac{s}{m}\beta_s\,dR/R\not=0$, otherwise with $\beta_{s-1}^*=\beta_{s-1}/\beta_s$ one finds that $d\tau=\frac 1s d\beta_{s-1}^*$, which, as above, contradicts the transcendence of $\tau$. So the assertion holds in case (ii) for $s>1$. There remains to consider the cases $s=1$ or $s\leq 0$.
\end{itemize}
\item We turn to the degree zero term in \eqref{elemrel}, including scenarios with $v=0$.  We first show that the
$u_i$ and $v$ can be chosen such that $r_i = 0$, with $\alpha_0 = 1$, and $s<0$.
\begin{itemize}
\item In Case (i) we find
\begin{equation}\label{zeroi}
\left(\sum (c_i\dfrac{d\alpha_i}{\alpha_i}+r_i\,dR)+d\beta_0\right)\wedge\Omega=0.
\end{equation}
If $\omega:=\sum (c_i\dfrac{d\alpha_i}{\alpha_i}+r_i\,dR)+d\beta_0\not=0$ then we are done. If $\omega=0$, then set
\[
\widehat u_i:=u_i/(\alpha_i\tau^{r_i})\quad\text{and} \quad \widehat v:=v-\beta_0.
\]
From
\[
d\widehat u_i=\tau^{-r_i}\alpha_i^{-1}\left(du_i-u_i(\dfrac{d\alpha_i}{\alpha_i}+\frac{r_i}{m} \,dR)\right)
\]

one sees
\[
\dfrac{d\widehat u_i}{\widehat u_i}=\dfrac{du_i}{u_i}-\left(\dfrac{d\alpha_i}{\alpha_i}+\frac{r_i}{m}\, dR\right),
\]
for all $i$, which, in view of $\omega=0$, implies
\begin{equation}\label{hateq1}
  \left(\sum c_i\,\dfrac{d\widehat u_i}{\widehat u_i} +d\widehat v\right) = \left(\sum c_i\,\dfrac{d u_i}{u_i} +d v\right),
\end{equation}
and hence,
\begin{equation}\label{hateq}
  \left(\sum c_i\,\dfrac{d\widehat u_i}{\widehat u_i} +d\widehat v\right)\wedge\Omega=0, \quad \sum c_i\,\dfrac{d\widehat u_i}{\widehat u_i} +d\widehat v\not=0.
\end{equation}
We will continue this discussion below.
\item  In Case (ii) we find
\begin{equation}\label{zeroi}
\left(\sum c_i\dfrac{d\alpha_i}{\alpha_i}+d\beta_0+\beta_1\,dR/R\right)\wedge\Omega=0,
\end{equation}
and by the discussion in item 1 we may assume that $\beta_1$ is constant (possibly zero).  As in case (i), if $\omega:=\sum c_i\dfrac{d\alpha_i}{\alpha_i}+d\beta_0+\beta_1\,dR/R\not=0$ then we are done. If $\omega=0$, then set
\[
\widehat u_i:=u_i/(\alpha_i\tau^{r_i})\quad\text{and} \quad \widehat v:=v-\beta_0-t\beta_1,
\]
such that \eqref{hateq1} holds, and hence \eqref{hateq}, with the $u_i$ and $v$ of the desired form.
\end{itemize}
\item  We can therefore assume that the expansions of all the $u_i$ have
$r_i = 0$ with $\alpha_i=1$, and that $s<0$ in the expansion of $v$. 
If $z=1+\tau z^*$, with $z^*\in \widetilde L_0[[\tau]]$ then, using the formal expansion for
$\log(1+x)$, we can find $\lambda \in \tau^{-1}\,\widetilde L_0[[\tau^{-1}]]$, such that $d\lambda=dz/z$. Doing this for each 
$u_i$, we can find an expression, 
\[\phi = \sum \phi_k \tau^k\in \tau^{-1}\,\widetilde L_0[[\tau^{-1}]],
\] such that 
\[
d\phi=\sum c_i\,\dfrac{d u_i}{u_i} +d v \not=0,\quad d\phi\wedge \Omega =0.
\]
\begin{itemize}
\item In Case (i), there exists some $\ell \in \mathbb{Z}$ such that
\[
0\not=d(\tau^\ell\phi_\ell)=\tau^\ell\left(\dfrac{d\phi_\ell}{\phi_\ell}+\frac{\ell}{m}\, dR\right)
\]
and
\[
\left(\dfrac{d\phi_\ell}{\phi_\ell}+\frac{\ell}{m}\, dR\right)\wedge \Omega=0,
\]
whence the assertion follows.
\item In Case (ii), we consider the highest index $\ell$ such that $\phi_\ell\not=0$. Then $d\phi_{\ell}\wedge\Omega=0$, and in case of constant $\phi_\ell$ the terms in $\tau^{\ell-1}$ give
\[
\left(d\phi_{\ell-1}+\frac{\ell}{m}\phi_\ell\,dR/R\right)\wedge\Omega=0.
\]
The assertion follows since $d\phi_{\ell-1}+\frac{\ell}{m}\phi_\ell\,dR/R\not =0$:
Otherwise, we have $\tau = c-\frac{\phi_{\ell-1}}{\ell\phi_\ell} \in \widetilde{K}$ for some constant $c$, which is a contradiction.
\end{itemize}
\end{enumerate}
\end{proof}

\section{Appendix}

\subsection{Laurent and Puiseux expansions}
Here we collect some pertinent facts about power series expansions. Both Lemma \ref{laurentlem} and Lemma \ref{puislem} might be considered standard. But we include them (with proof sketches), for easy reference, and because they are crucial for our arguments. {The first lemma implies the existence of a formal Laurent expansion, as used in the proof of Theorem \ref{MSing}.}
\begin{lemma}\label{laurentlem}
Let $L_0$ be a field, $L=L_0(t)$ with $t$ transcendental over $L_0$, and $r\in L$ nonzero. Then there exist an integer {$M$} and $c_j\in L_0$, $j\geq 0$, so that for any integer $\ell\geq 0$ there exists $r_{\ell+1}\in L$ with either $r_{\ell+1}=0$ or $r_{\ell+1}(0)\not=0$ such that 
{
\begin{equation*}
    t^M r=c_0+tc_1+\cdots+t^\ell c_\ell+t^{\ell+1}r_{\ell+1}.
\end{equation*}
}
Moreover, the assertion also holds with $t$ replaced by $t^{-1}$.
Mutatis mutandis, these statements also hold for elements of any finite dimensional vector space over $L$.
\end{lemma}
\begin{proof}{
There is an integer $M$ such that
\begin{equation*}
    t^Mr=\dfrac{a_0+ta_1+\cdots}{b_0+tb_1+\cdots}\text{  with  } a_0\not=0, b_0\not=0.
\end{equation*}
To determine the $c_j$, proceed recursively, starting with $c_0= a_0/b_0$ and either $t^Mr$ constant or
\begin{equation*}
   t^M r-\dfrac{a_0}{b_0}=\dfrac{a_0+ta_1+\cdots-a_0/b_0(b_0+tb_1+\cdots)}{b_0+tb_1+\cdots}= t^{M_1} s_1;\quad M_1\geq 1\text{  and  }s_1(0)\not=0.
\end{equation*}
In case $M_1=1$ set $r_1=s_1$; otherwise set $\ell=M_1$, $r_1=\cdots=r_{\ell-1}=0$ and $r_\ell=s_1$.
}
The recursion step works by applying the same argument to $r_\ell$.\\ The last assertion is immediate from $L_0(t)=L_0(t^{-1})$.
\end{proof}
The following lemma is a consequence of the Newton-Puiseux theorem; see Abhyankar \cite{abhyankar}, Lecture 12. We cannot directly use the theorem as stated in \cite{abhyankar}, {since it assumes an algebraically closed base field. However, a closer look
at the proof shows that only a finite extension of the
base field is needed.  Our sketch will closely trace Abhyankar's proof, just adding the necessary modifications to his arguments. Readers who are interested in details are asked to combine our outline with the reference \cite{abhyankar}.}
\begin{lemma}\label{puislem}
Let $L_0$ be  field of characteristic zero, $t$ transcendental over $L_0$, moreover let $q$ be algebraic over $L_0(t)$, and $L=L_0(t,q)$. Then there exist
a finite algebraic extension $\widetilde L_0$ of $L_0$
   and a positive integer $m$,
such that every element of $L$ admits a representation 
\begin{equation*}
    \sum_{i=N}^{\infty} a_i\tau^i;\quad \tau=t^{1/m},
\end{equation*}
with all $a_i\in \widetilde L_0((t))$.
Moreover, this statement also holds for all elements of any finite dimensional vector space over $L$, and analogous statements hold with $\tau$ replaced by $\tau^{-1}$.
\end{lemma}
\begin{proof}[{Sketch of proof}]
    {It suffices to prove the statement for $q$, since $L=L_0(t)\,[q]$, and with $q$, every polynomial in $q$ with coefficients in $L_0(t)$ will have a representation in $\widetilde L_0((\tau))$ as asserted.\\
    Let $ Q(t,y)\in L_0(t)[y]$ denote the minimal polynomial of $q$ over $L_0(t)$;
    \begin{equation*}
    Q(t,Y)=Y^n+c_1 Y^{n-1}+\cdots+c_n,
    \end{equation*}
    with all $c_j\in L_0(t)\subset L_0((t))$.
    Due to Lemma \ref{laurentlem} we may assume that $n>1$. We will show the existence of a finite extension $\widehat L_0$ of $L_0$ such that $Q$ is reducible over $\widehat L_0((\tau))$. The following arguments (due to Abhyankar) do not rely on rationality of the $c_j$, or irreducibility of $Q$.\\
    In case $Q=Y^n$ reducibility is obvious. Otherwise, following Abhyankar's proof, there exists a rational number $d = k/m$ with $m>0$, so that, taking $\tau=t^{1/m}$ and $y = \tau^k(Y+c_1/n)$, we can define 
    \[
    \widehat Q(\tau,y):=\tau^{kn}\,Q(t,Y)
    =y^n+\sum_{j=1}^n \widehat c_j(\tau)\,y^{n-j}
    \]
    with all $\widehat c_j\in L_0[[\tau]]$, with $\widehat c_1=0$ and $\widehat c_j(0)\not=0$ for some $j\neq 1$. \\
    Now set $\widehat Q_0:=\widehat Q(0,y)$, and let $\widehat L_0$ be its splitting field over $L_0$. By the argument in \cite{abhyankar}, p.~93, one has
    \[
    \widehat Q_0= \widehat P_{0,1}\cdot \widehat P_{0,2},
    \]
    with relatively prime $\widehat P_{0,i}\in \widehat L_0[y]$. By Hensel's lemma (the proof in \cite{abhyankar}, p.~90, works for all fields of
    characteristic 0) one gets
    \[
    \widehat Q =\widehat P_1\cdot \widehat P_2
    \]
    with relatively prime $\widehat P_i\in \widehat L_0[[\tau]]\,[y]$.
    By the correspondence between $Q$ and $\widehat Q$ one arrives at 
    \[
    Q=P_1\cdot P_2; \quad P_i\in \widehat L_0((\tau))\,[y].
    \]
    Proceeding by induction on the degree (possibly requiring further field extensions and increase of $m$) one obtains 
   a finite field extension $\widetilde L_0$, and a decomposition 
    \[
    Q(t,y)=\prod(y-\eta_j)\,,
    \]
    as a product of linear factors, with the $\eta_j\in \widetilde L_0((t^{1/\tilde{m}}))$. 
    Now $Q(t,q)=0$ shows that $q=\eta_k $ for some $k$. }
    To prove the assertion for decreasing powers of $\tau$, start with $s=t^{-1}$ and repeat the argument over $L_0(s)$.\\
    The generalization to finite dimensional vector spaces over $L$ is straightforward.
\end{proof}

\section*{Acknowledgments}
C. Pantazi is partially supported by the grants PID2024-158570NB-I00, 2025 SGR 00447 and
 PID-2021-122954NB-100 funded by MCIN/AEI/ 10.13039/501100011033 and by ``ERDF A way of making Europe"
 
%%%%%%%%%%%%%%%%%%%%%%%%%%%%%%%%%%%%%%%%%%%%%%%%%%%%%%%%%%%%%%%%%%%%%%%%%%%%%%%%%%%%%%%

\end{document}